\documentclass[letterpaper,11pt]{article} 
\usepackage{fullpage}

\usepackage{amsmath,amsfonts, amssymb}
\usepackage{amsthm}
\usepackage{graphicx}
\usepackage{algorithm}
\usepackage{algpseudocode}
\usepackage{comment}
\usepackage{url}

\newtheorem{theorem}{Theorem}
\newtheorem{lemma}[theorem]{Lemma}
\newtheorem{definition}[theorem]{Definition}

\newtheorem{corollary}[theorem]{Corollary}

\begin{document}

\title{On Finding Quantum Multi-collisions}
\author{Qipeng Liu and Mark Zhandry\\ Princeton University}


\maketitle

\begin{abstract}
    A $k$-collision for a compressing hash function $H$ is a set of $k$ distinct inputs that all map to the same output.  In this work, we show that for any constant $k$,  $\Theta\left(N^{\frac{1}{2}(1-\frac{1}{2^k-1})}\right)$ quantum queries are both necessary and sufficient to achieve a $k$-collision with constant probability.  This improves on both the best prior upper bound (Hosoyamada et al., ASIACRYPT 2017) and provides the first non-trivial lower bound, completely resolving the problem.
\end{abstract}

\section{Introduction}

Collision resistance is one of the central concepts in cryptography.  A \emph{collision} for a hash function $H:\{0,1\}^m\rightarrow\{0,1\}^n$ is a pair of distinct inputs $x_1\neq x_2$ that map to the same output: $H(x_1)=H(x_2)$.  

\paragraph{Multi-collisions.} Though receiving comparatively less attention in the literature, multi-collision resistance is nonetheless an important problem.  A $k$-collision for $H$ is a set of $k$ distinct inputs $\{x_1,\dots,x_k\}$ such that $x_i\neq x_j$ for $i\neq j$ where $H(x_i)=H(x_j)$ for all $i,j$.  

Multi-collisions frequently surface in the analysis of hash functions and other primitives.  Examples include MicroMint~\cite{SPW:RivSha97}, RMAC~\cite{FSE:JauJouVal02}, chopMD~\cite{FSE:ChaNan08b}, Leamnta-LW~\cite{ICISC:HIKOPY10}, PHOTON and Parazoa~\cite{SCN:NaiOht14}, the Keyed-Sponge~\cite{AC:JovLuyMen14}, all of which assume the multi-collision resistance of a certain function.  Multi-collisions algorithms have also been used in attacks, such as the MDC-2~\cite{EC:KMRT09},  HMAC~\cite{IWSEC:NSWY13}, Even-Mansour~\cite{AC:DDKS14}, and LED~\cite{FSE:NikWanWu13}.  Multi-collision resistance for polynomial $k$ has also recently emerged as a theoretical way to avoid keyed hash functions~\cite{STOC:BitKalPan18,EC:BDRV18}, or as a  useful cryptographic primitives, for example, to build statistically hiding commitment schemes with succinct interaction\cite{EC:KomNaoYog18}.

\paragraph{Quantum.} Quantum computing stands to fundamentally change the field of cryptography.  Importantly for our work, Grover's algorithm~\cite{STOC:Grover96} can speed up brute force searching by a quadratic factor, greatly increasing the speed of pre-image attacks on hash functions.  In turn, Grover's algorithm can be used to find ordinary collisions ($k=2$) in time $O(2^{n/3})$, speeding up the classical ``birthday'' attack which requires $O(2^{n/2})$ time.  It is also known that, in some sense (discussed below), these speedups are optimal~\cite{AS2004,Zhandry15}.  These attacks require updated symmetric primitives with longer keys in order to make such attacks intractable.

\subsection{This Work: Quantum Query Complexity of Multi-collision Resistance}

In this work, we consider \emph{quantum} multi-collision resistance.  
Unfortunately, little is known of the difficulty of finding \emph{multi}-collisions for $k\geq 3$ in the quantum setting.  The only prior work on this topic is that of Hosoyamada et al.~\cite{AC:HosSasXag17}, who give a $O(2^{4n/9})$ algorithm for 3-collisions, as well as algorithms for general constant $k$.  On the lower bounds side, the $\Omega(2^{n/3})$ from the $k=2$ case applies as well for higher $k$, and this is all that is known.

We completely resolve this question, giving tight upper and lower bounds for any constant $k$.  In particular, we consider the quantum query complexity of multi-collisions.  We will model the hash function $H$ as a \emph{random oracle}.  This means, rather than getting concrete code for a hash function $H$, the adversary is given black box access to a function $H$ chosen uniformly at random from the set of all functions from $\{0,1\}^m$ into $\{0,1\}^n$.  Since we are in the quantum setting, black box access means the adversary can make quantum queries to $H$.  Each query will cost the adversary 1 time step.  The adversary's goal is to solve some problem --- in our case find a $k$-collision --- with the minimal cost.  Our results are summarized in Table~\ref{table1}.  Both our upper bounds and lower bounds improve upon the prior work for $k\geq 3$; for example, for $k=3$, we show that the quantum query complexity is $\Theta(2^{3n/7})$.

\begin{table}
	\centering
\begin{tabular}[h]{|c|c|c|}\hline
			& Upper Bound (Algorithm)			& Lower Bound 						\\\hline
\cite{BHT98}            & $O(2^{m/3})$ for $k=2$ (2-to-1)		&             						\\\hline
\cite{AS2004}    &                               		& $\Omega(2^{m/3})$ for $k=2$ (2-to-1)			\\\hline
\cite{Zhandry15}	& $O(2^{n/3})$ for $k=2$ (Random, $m\geq n/2)$	& $\Omega(2^{n/3})$ for $k=2$ (Random) 			\\\hline
\cite{AC:HosSasXag17}	& $O\left(2^{\frac{1}{2}(1-\frac{1}{3^{k-1}})n}\right)$ ($m\geq n+\log k$)	& 							\\\hline
{\bf This Work}		& $\mathbf{O\left(2^{\frac{1}{2}(1-\frac{1}{2^k-1})n}\right)}$ {\bf ($\mathbf{m\geq n+\log k}$)}	& $\mathbf{\Omega\left(2^{\frac{1}{2}(1-\frac{1}{2^k-1})n}\right)}$ {\bf (Random)} 	\\\hline
\end{tabular}

	\caption{\label{table1} Quantum query complexity results for $k$-collisions.  $k$ is taken to be a constant, and all Big $O$ and $\Omega$ notations hide constants that depend on $k$.  In parenthesis are the main restrictions for the lower bounds provided.  We note that in the case of 2-to-1 functions, $m\leq n+1$, so implicitly these bounds only apply in this regime.  In these cases, $m$ characterizes the query complexity.  On the other hand, for random or arbitrary functions, $n$ is the more appropriate way to measure query complexity.  We also note that for arbitrary functions, when $m\leq n+\log (k-1)$, it is possible that $H$ contains no $k$-collisions, so the problem becomes impossible.  Hence, $m\geq n+\log k$ is essentially tight.  For random functions, there will be no collisions w.h.p unless $m\gtrsim (1-\frac{1}{k})n$, so algorithms on random functions must always operate in this regime.}
\end{table}

\subsection{Motivation}

Typically, the parameters of a hash function are set to make finding collisions intractable.  One particularly important parameter is the output length of the hash function, since the output length in turn affects storage requirements and the efficiency of other parts of a cryptographic protocol.  

Certain attacks, called \emph{generic} attacks, apply regardless of the implementation details of the hash function $H$, and simply work by evaluating $H$ on several inputs.  For example, the birthday attack shows that it is possible to find a collision in time approximately $2^{n/2}$ by a classical computer.  Generalizations show that $k$-collisions can be found in time $\Theta(2^{(1-1/k)n})$\enspace\footnote{Here, the Big Theta notation hides a constant that depends on $k$}.  

These are also known to be optimal among \emph{classical} generic attacks.  This is demonstrated by modeling $H$ as an oracle, and counting the number of queries needed to find ($k$-)collisions in an arbitrary hash function $H$.  In cryptographic settings, it is common to model $H$ as a \emph{random} function, giving stronger \emph{average case} lower bounds.

Understanding the effect of generic attacks is critical.  First, they cannot be avoided, since they apply no matter how $H$ is designed.  Second, other parameters of the function, such as the number of iterations of an internal round function, can often be tuned so that the best known attacks are in fact generic.  Therefore, for many hash functions, the complexity of generic attacks accurately represents the actual cost of breaking them.

Therefore, for ``good'' hash functions where generic attacks are optimal, in order to achieve security against classical adversaries $n$ must be chosen so that $t=2^{n/2}$ time steps are intractable.  This often means setting $t= 2^{128}$, so $n=256$.  In contrast, generic classical attacks can find $k$-collisions in time $\Theta(2^{(1-1/k)n})$.  For example, this means that $n$ must be set to $192$ to avoid $3$-collisions, or $171$ to avoid $4$-collisions.

Once quantum computers enter the picture, we need to consider quantum queries to $H$ in order to model actual attacks that evaluate $H$ in superposition.  This changes the query complexity, and makes proving bounds much more difficult.  Just as understanding query complexity in the classical setting was crucial to guide parameter choices, it will be critical in the quantum world as well.  

\medskip

We also believe that quantum query complexity is an important study in its own right, as it helps illuminate the effects quantum computing will have on various areas of computer science.  It is especially important to cryptography, as many of the questions have direct implications to the post-quantum security of cryptosystems.  Even more, the techniques involved are often closely related to proof techniques in post-quantum cryptography.  For example, bounds for the quantum query complexity of finding collisions in random functions~\cite{Zhandry15}, as well as more general functions~\cite{EPRINT:EbrUnr17,EPRINT:BalEatSon17}, were developed from techniques for proving security in the quantum random oracle model~\cite{AC:BDFLSZ11,C:Zhandry12,TCC:TarUnr16}.  Similarly, the lower bounds in this work build on techniques for proving quantum indifferentiability~\cite{EPRINT:Zhandry18}.  On the other hand, proving the security of MACs against superposition queries~\cite{EC:BonZha13} resulted in new lower bounds for the quantum oracle interrogation problem~\cite{FOCS:VanDam98} and generalizations~\cite{zhandry2015quantum}.

Lastly, multi-collision finding can be seen as a variant of $k$-distinctness, which is essentially the problem of finding a $k$-collision in a function $H:\{0,1\}^n\rightarrow\{0,1\}^n$, where the $k$-collision may be unique and all other points are distinct.  The quantum query complexity of $k$-distinctness is currently one of the main open problems in quantum query complexity.  An upper bound of $(2^n)^{\frac{3}{4}-\frac{1}{4(2^k-1)}}$ was shown by Belovs~\cite{FOCS:Belovs12}.  The best known lower bound is $\Omega((2^n)^{\frac{3}{4}-\frac{1}{2k}})$~\cite{STOC:BunKotTha18}.  Interestingly, the dependence of the exponent on $k$ is exponential for the upper bound, but polynomial for the lower bound, suggesting a fundamental gap our understanding of the problem.

Note that our results do not immediately apply in this setting, as our algorithm operates only in a regime where there are many ($\leq k$-)collisions, whereas $k$-distinctness applies even if the $k$-collision is unique and all other points are distinct (in particular, no $(k-1)$-collisions).  On the other hand, our lower bound is always lower than $2^{n/2}$, which is trivial for this problem.  Nonetheless, both problems are searching for the same thing --- namely a $k$-collisions --- just in different settings.  We hope that future work may be able to extend our techniques to solve the problem of $k$-distinctness.

\subsection{The “Reciprocal Plus 1” Rule}  For many search problems over random functions, such as pre-image search, collision finding, $k$-sum, quantum oracle interrogation, and more, a very simple folklore rule of thumb translates the classical query complexity into quantum query complexity.  

In particular, let $N = 2^n$,  all of these problems have a classical query complexity $\Theta(N^{1/\alpha})$ for some rational number $\alpha$.  Curiously, the quantum query complexity of all these problems is always $\Theta(N^{\frac{1}{\alpha+1}})$.  

In slightly more detail, for all of these problems the best classical $q$-query algorithm solves the problem with probability $\Theta(q^c/N^d)$ for some constants $c,d$.  Then the classical query complexity is $\Theta(N^{d/c})$.  For this class of problems, the success probability of the best $q$ query quantum algorithm is obtained simply by increasing the power of $q$ by $d$.  This results in a quantum query complexity of $\Theta(N^{d/(c+d)})$.  Examples:
\begin{itemize}
	\item Grover's pre-image search~\cite{STOC:Grover96} improves success probability from $q/N$ to $q^2/N$, which is known to be optimal~\cite{BBBV97}.  The result is a query complexity improvement from $N=N^{1/1}$ to $N^{1/2}$.
	
	Similarly, finding, say, 2 pre-images has classical success probability $q^2/N^2$; it is straightforward to adapt known techniques to prove that the best quantum success probability is $q^4/N^2$.  Again, the query complexity goes from $N$ to $N^{1/2}$.  Analogous statements hold for any \emph{constant} number of pre-images.
	\item The BHT collision finding algorithm~\cite{BHT98} finds a collision with probability $q^3/N$, improving on the classical birthday attack $q^2/N$.  Both of these are known to be optimal~\cite{AS2004,Zhandry15}.  Thus quantum algorithms improve the query complexity from $N^{1/2}$ to $N^{1/3}$.
	
	Similarly, finding, say, 2 distinct collisions has classical success probability $q^4/N^2$, whereas we show that the quantum success probability is $q^6/N^2$.  More generally, any \emph{constant} number of distinct collisions conforms to the Reciprocal Plus 1 Rule.  
	
	\item $k$-sum asks to find a set of $k$ inputs such that the sum of the outputs is 0.  This is a different generalization of collision finding than what we study in this work.  Classically, the best algorithm succeeds with probability $q^k/N$.  Quantumly, the best algorithm succeeds with probability $q^{k+1}/N$~\cite{ITCS:BelSpa13,EPRINT:Zhandry18}.  Hence the query complexity goes from $N^{1/k}$ to $N^{1/(k+1)}$.
	
	Again, solving for any \emph{constant} number of distinct $k$-sum solutions also conforms to the Reciprocal Plus 1 Rule.
	\item In the oracle interrogation problem, the goal is to compute $q+1$ input/output pairs, using only $q$ queries.  Classically, the best success probability is clearly $1/N$.  Meanwhile, Boneh and Zhandry~\cite{EC:BonZha13} give a quantum algorithm with success probability roughly $q/N$, which is optimal.
\end{itemize}

Some readers may have noticed that Reciprocal Plus 1 (RP1) rule does not immediately appear to apply the Element Distinctness.  The Element Distinctness problem asks to find a collision in $H:[M]\rightarrow[N]$ where the collision is unique.  Classically, the best algorithm succeeds with probability $\Theta(q^2/M^2)$.  On the other hand, quantum algorithms can succeed with probability $\Theta(q^3/M^2)$, which is optimal~\cite{FOCS:Ambainis04,Zhandry15}.  This does not seem to follow the prediction of the RP1 rule, which would have predicted $q^4/M^2$.  However, we note that unlike the settings above which make sense when $N\ll M$, and where the complexity is characterized by $N$, the Element Distinctness problem requires $M\leq N$ and the complexity is really characterized by the domain size $M$.  Interestingly, we note that for a random expanding function, when $N\approx M^2$, there will with constant probability be exactly one collision in $H$.  Thus, in this regime the collision problem matches the Element Distinctness problem, and the RP1 rule gives the right query complexity!

Similarly, the quantum complexity for $k$-sum is usually written as $M^{k/(k+1)}$, not $N^{1/(k+1)}$.  But again, this is because most of the literature considers $H$ for which there is a unique $k$-sum and $H$ is non-compressing, in which case the complexity is better measured in terms of $M$.  Notice that a random function will contain a unique $k$ collision when $N\approx M^k$, in which case the bound we state (which follows the RP1 rule) exactly matches the statement usually given.  

On the other hand, the RP1 rule does not give the right answer for $k$-distinctness for $k\geq 3$, since the RP1 rule would predict the exponent to approach $1/2$ for large $k$, whereas prior work shows that it approaches $3/4$ for large $k$.  That RP1 does not apply perhaps makes sense, since there is no setting of $N,M$ where a random function will become an instance of $k$-distinctness: for any setting of parameters where a random function has a $k$-collision, it will also most likely have many $(k-1)$-collisions.   

The takeaway is that the RP1 Rule seems to apply for natural search problems that make sense on random functions when $N\ll M$.  Even for problems that do not immediately fit this setting such as Element Distinctness, the rule often still gives the right query complexity by choosing $M,N$ so that a random function is likely to give an instance of the desired problem.

\paragraph{Enter $k$-collisions.}  In the case of $k$-collisions, the classical best success probability is $q^k/N^{(k-1)}$, giving a query complexity of $N^{(k-1)/k}=N^{1-1/k}$.  Since the $k$-collision problem is a generalization of collision finding, is similar in spirit to the problems above, and applies to compressing random functions, one may expect that the Reciprocal Plus 1 Rule applies.  If true, this would give a quantum success probability of $q^{2k-1}/N^{k-1}$, and a query complexity of $N^{(k-1)/(2k-1)}=N^{\frac{1}{2}(1-\frac{1}{2k-1})}$.

Even more, for small enough $q$, it is straightforward to find a $k$-collision with probability $O(q^{2k-1}/N^{k-1})$ as desired.  In particular, divide the $q$ queries into $k-1$ blocks.  Using the first $q/(k-1)$ queries, find a 2-collision with probability $(q/(k-1))^3/N=O(q^3/N)$.  Let $y$ be the image of the collision.  Then, for each of the remaining $(k-2)$ blocks of queries, find a pre-image of $y$ with probability $(q/(k-1))^2/N=O(q^2/N)$ using Grover search.  The result is $k$ colliding inputs with probability $O(q^{3+2(k-2)}/N^{k-1})=O(q^{2k-1}/N^{k-1})$.  It is also possible to prove that this is a lower bound on the success probability (see lower bound discussion below).  Now, this algorithm works as long $q\leq N^{1/3}$, since beyond this range the 2-collision success probability is bounded by $1<q^3/N$.  Nonetheless, it is asymptotically tight in the regime for which it applies.  This seems to suggest that the limitation to small $q$ might be an artifact of the algorithm, and that a more clever algorithm could operate beyond the $N^{1/3}$ barrier.  In particular, this strongly suggests $k$-collisions conforms to the Reciprocal Plus 1 Rule.

Note that the RP1 prediction gives an exponent that depends polynomially on $k$, asymptotically approaching $1/2$.  In contrast, the prior work of~\cite{AC:HosSasXag17} approaches $1/2$ \emph{exponentially fast} in $k$.  Thus, prior to our work we see an exponential vs polynomial gap for $k$-collisions, similar to the case of $k$-distinctness.  

\medskip

Perhaps surprisingly given the above discussion\footnote{At least, the authors found it surprising!}, our work demonstrates that the right answer is in fact exponential, refuting the RP1 rule for $k$-collisions.  

As mentioned above, our results do not immediately give any indication for the query complexity of $k$-distinctness.  However, our results may hint that $k$-distinctness also exhibits an exponential dependence on $k$.  We hope that future work, perhaps building on our techniques, will be able to resolve this question.

\subsection{Technical Details}

\subsubsection{The Algorithm}  At their heart, the algorithms for pre-image search, collision finding, $k$-sum, and the recent algorithm for $k$-collision, all rely on Grover's algorithm.  Let $f:\{0,1\}^n\rightarrow\{0,1\}$ be a function with a fraction $\delta$ of accepting inputs.  Grover's algorithm finds the input with probability $O(\delta q^2)$ using $q$ quantum queries to $f$.  Grover's algorithm finds a pre-image of a point $y$ in $H$ by setting $f(x)$ to be 1 if and only if $H(x)=y$.

The BHT algorithm~\cite{BHT98} uses Grover's to find a collision in $H$.  First, it queries $H$ on $q/2=O(q)$ random points, assembling a database $D$.  As long as $q\ll N^{1/2}$, all the images in $D$ will be distinct.  Now, it lets $f(x)$ be the function that equals 1 if and only if $H(x)$ is found amongst the images in $D$, and $x$ is not among the pre-images.  By finding an accepting input to $f$, one immediately finds a collision.  Notice that the fraction of accepting inputs is approximately $q/N$.  

By running Grover's for $q/2=O(q)$ steps, one obtains a such a pre-image, and hence a collision, with probability $O((q/N) q^2)=O(q^3/N)$.

Hosoyamada et al. show how this idea can be recursively applied to find multi-collisions.  For $k=3$, the first step is to find a database $D_2$ consisting of $r$ distinct 2-collisions.  By recursively applying the BHT algorithm, each 2-collision takes time $N^{1/3}$.  Then, to find a 3 collision, set up $f$ as before: $f(x)=1$ if and only if $H(x)$ is amongst the images in $D$ and $x$ is not among the pre-images.  The fraction of accepting inputs is approximately $r/N$, so Grover's algorithm will find a 3-collision in time $(N/r)^{1/2}$.  Setting $r$ to be $N^{1/9}$ optimizes the total query count as $N^{4/9}$.  For $k=4$, recursively build a table $D_3$ of 3-collisions, and set up $f$ to find a collision with the database.

\medskip

\noindent The result is an algorithm for  $k$-collisions for any constant $k$, using $O(N^{\frac{1}{2}(1-\frac{1}{3^{k-1}})})$ queries.  

\medskip

Our algorithm improves on Hosoyamada et al.'s, yielding a query complexity of $O(N^{\frac{1}{2}(1-\frac{1}{2^{k}-1})})$.  Note that for Hosoyamada et al.'s algorithm, when constructing $D_{k-1}$, many different $D_{k-2}$ databases are being constructed, one for each entry in $D_{k-1}$.  Our key observation is that a single database can be re-used for the different entries of $D_{k-1}$.  This allows us to save on some of the queries being made.  These extra queries can then be used in other parts of the algorithm to speed up the computation.  By balancing the effort correctly, we obtain our algorithm.  Put another way, the cost of finding many ($k$-)collisions can be amortized over many instances, and then recursively used for finding collisions with higher $k$.  Since the recursive steps involve solving many instances, this leads to an improved computational cost.

In more detail, we iteratively construct databases $D_1,D_2,\dots,D_{k}$.  Each $D_i$ will have $r_i$ $i$-collisions.  We set $r_k=1$, indicating that we only need a single $k$-collision.  To construct database $D_1$, simply query on $r_1$ arbitrary points.  To construct database $D_i,i\geq 2$, define the function $f_i$ that accepts inputs that collide with $D_{i-1}$ but are not contained in $D_{i-1}$.  The fraction of points accepted by $f_i$ is approximately $r_{i-1}/N$.  Therefore, Grover's algorithm returns an accepting input in time $(N/r_{i-1})^{1/2}$.  We simply run Grover's algorithm $r_i$ times \emph{using the same database} $D_{i-1}$ to construct $D_i$ in time $r_i(N/r_{i-1})^{1/2}$.

Now we just optimize $r_1,\dots,r_{k-1}$ by setting the number of queries to construct each database to be identical.  Notice that $r_1=O(q)$, so solving for $r_i$ gives us 
\[r_k=O\left(\frac{q^{2-\frac{1}{2^{k-1}}}}{N^{1-\frac{1}{2^{k-1}}}} \right)\]

Setting $r_k=1$ and solving for $q$ gives the desired result.  In particular, in the case $k=3$, our algorithm finds a collision in time $O(N^{3/7})$.

\subsubsection{The Lower Bound.} Notice that our algorithm fails to match the result one would get by applying the ``Reciprocal Plus 1 Rule''.  Given the discussion above, one may expect that our iterative algorithm could potentially be improved on even more.  To the contrary we prove that, in fact, our algorithm is asymptotically optimal for any constant $k$.

Toward that end, we employ a recent technique developed by Zhandry~\cite{EPRINT:Zhandry18} for analyzing quantum queries to \emph{random} functions.  We use this technique to show that our algorithm is tight for random functions, giving an average-case lower bound.

\paragraph{Zhandry's ``Compressed Oracles.''}  Zhandry demonstrates that the information an adversary knows about a random oracle $H$ can be summarized by a database $D^*$ of input/output pairs, which is updated according to special rules.  In Zhandry's terminology, $D^*$ is the ``compressed standard/phase oracle''.  

This $D^*$ is not a classical database, but technically a superposition of databases, meaning certain amplitudes are assigned to each possible database.  $D^*$ can be measured, obtaining an actual classical database $D$ with probability equal to its amplitude squared.  In the following discussion, we will sometimes pretend that $D^*$ is actually a classical database.  While inaccurate, this will give the intuition for the lower bound techniques we employ.  In the section \ref{lowerboundsection} we take care to correctly analyze $D^*$ as a superposition of databases.  
\medskip

\noindent Zhandry shows roughly the following:
\begin{itemize}
	\item Consider any ``pre-image problem'', whose goal is to find a set of pre-images such that the images satisfy some property.  For example, $k$-collision is the problem of finding $k$ pre-images such that the corresponding images are all the same.  
	
	Then after $q$ queries, consider measuring $D^*$.  The adversary can only solve the pre-image problem after $q$ queries if the measured $D^*$ has a solution to the pre-image problem.

Thus, we can always upper bound the adversary's success probability by upper bounding the probability $D^*$ contains a solution.

	\item $D^*$ starts off empty, and each query can only add one point to the database.  

	\item For any image point $y$, consider the  amplitude on databases containing $y$ as a function of $q$ (remember that amplitude is the square root of the probability).  Zhandry shows that this amplitude can only increase by $O(\sqrt{1/N})$ from one query to the next.  More generally, for a set $S$ of $r$ different images, the amplitude on databases containing any point in $S$ can only increase by $O(\sqrt{|S|/N})$.  
\end{itemize}

The two results above immediately imply the optimality of Grover's search.  In particular, the amplitude on databases containing $y$ is at most $O(q \sqrt{1/N})$ after $q$ queries, so the probability of obtaining a solution is the square of this amplitude, or $O(q^2/N)$.  This also readily gives a lower bound for the collision problem.  Namely, in order to introduce a collision to $D^*$, the adversary must add a point that collides with one of the existing points in $D^*$.  Since there are at most $q$ such points, the amplitude on such $D^*$ can only increase by $O(\sqrt{q/N})$.  This means the overall amplitude after $q$ queries is at most $O(q^{3/2}/N^{1/2})$.  Squaring to get a probability gives the correct lower bound.

\paragraph{A First Attempt.} Our core idea is to attempt a lower bound for $k$-collision by applying these ideas recursively.  The idea is that, in order to add, say, a 3-collision to $D^*$, there must be an existing 2-collision in the database.  We can then use the 2-collision lower bound to bound the increase in amplitude that results from each query.

More precisely, for very small $q$, we can bound the amplitude on databases containing $\ell$ distinct 2-collisions as $O(\;(q^{3/2}/N^{1/2})^\ell)$.  If $q\ll N^{1/3}$, $\ell$ must be a constant else this term is negligible.  So we can assume for $q<N^{1/3}$ that $\ell$ is a constant.

Then, we note that in order to introduce a 3-collision, the adversary's new point must collide with one of the existing 2-collisions.  Since there are at most $\ell$, we know that the amplitude increases by at most $O(\sqrt{\ell}/N^{1/2})=O(1/N^{1/2})$ since $\ell$ is a constant.  This shows that the amplitude on databases with 3-collisions is at most $q/N^{1/2}$.  

We can bound the amplitude increase even smaller by using not only the fact that the database contains at most $\ell$ 2-collisions, but the fact that the amplitude on databases containing even a single 2-collision is much less than 1.  In particular, it is $O(q^{3/2}/N^{1/2})$ as demonstrated above.  Intuitively, it turns out we can actually just multiply the $1/N^{1/2}$ amplitude increase in the case where the database contains a 2-collision by the $q^{3/2}/N^{1/2}$ amplitude on databases containing any 2-collision to get an overall amplitude increase of $q^{3/2}/N$.  

Overall then, we upper bound the amplitude after $q<N^{1/3}$ queries by $O(q^{5/2}/N)$, given an upper bound of $O(q^5/N^2)$ on the probability of finding a 3-collision.  This lower bound can be extended recursively to any constant $k$-collisions, resulting in a bound that exactly matches the Reciprocal Plus 1 Rule, as well as the algorithm for small $q$!  This again seems to suggest that our algorithm is not optimal.

\paragraph{Our Full Proof.} There are two problems with the argument above that, when resolved, actually do show our algorithm is optimal.  First, when $q\geq N^{1/3}$, the $O(q^{3/2}/N^{1/2})$ part of the amplitude bound becomes vacuous, as amplitudes can never be more than 1.  Second, the argument fails to consider algorithms that find many 2-collisions, which is possible when $q>N^{1/3}$.  Finding many 2-collisions of course takes more queries, but then it makes extending to 3-collisions easier, as there are more collisions in the database to match in each iteration.  

In our full proof, we examine the amplitude on the databases containing a 3-collision as well as $r$ 2-collisions, after $q$ queries.  We call this amplitude $g_{q,r}$.  We show a careful recursive formula for bounding $g$ using Zhandry's techniques, which we then solve.

More generally, for any constant $k$, we let $g^{(k)}_{q,r,s}$ be the amplitude on databases containing exactly $r$ distinct $(k-1)$-collisions and at least $s$ distinct $k$-collisions after $q$ queries.  We develop a multiply-recursive formula for the $g^{(k)}$ in terms of the $g^{(k)}$ and $g^{(k-1)}$.  We then recursively plug in our solution to $g^{(k-1)}$ so that the recursion is just in terms of $g^{(k)}$, which we then solve using delicate arguments. 

Interestingly, this recursive structure for our lower bound actually closely matches our algorithm.  Namely, our proof lower bounds the difficulty of adding an $i$-collision to a database $D^*$ containing many $i-1$ collisions, exactly the problem our algorithm needs to solve. Our techniques essentially show that every step of our algorithm is tight, resulting in a lower bound of $\mathbf{\Omega\left(N^{\frac{1}{2}(1-\frac{1}{2^k-1})}\right)}$, exactly matching our algorithm.  Thus, we solve the quantum query complexity of $k$-collisions.

\subsection{Other Related Work}
Most of the related work has been mentioned earlier. Recently, in 
\cite{cryptoeprint:2018:1122}, Hosoyamada, Sasaki, Tani and Xagawa gave the same improvement. And they also showed that, their algorithm can also find a multi-collision for a more general setting where $|X| \geq \frac{l}{c_N} \cdot |Y|$ for any positive value $c_N \geq 1$ which is in $o(N^{\frac{1}{2^l - 1}})$ and  find a multiclaw for random functions with the same query complexity. They also noted that our improved collision finding algorithm for the case $|X| \geq l \cdot |Y|$ was reported in the Rump Session of AsiaCrypt 2017. They did not give an accompanying lower bound.

\section*{Acknowledgements}

This work is supported in part by NSF. Opinions, findings and conclusions or recommendations expressed in this material are those of the author(s) and do not necessarily reflect the views of NSF.

\section{Preliminaries}

Here, we recall some basic facts about quantum computation, and review the relevant literature on quantum search problems.

\subsection{Quantum Computation}
A quantum system $Q$ is defined over a finite set $B$ of classical states. In this work we will consider $B = \{0,1\}^n$. A \textbf{pure state} over $Q$ is a unit vector in $\mathbb{C}^{|B|}$, which assigns a complex number to each element in $B$. In other words, let $|\phi\rangle$ be a pure state in  $Q$, we can write $|\phi\rangle$ as:
\begin{equation*}
    |\phi\rangle = \sum_{x \in B} \alpha_x |x\rangle
\end{equation*}
where $\sum_{x \in B} |\alpha_x|^2 = 1$ and $\{|x\rangle\}_{x \in B}$ is called 
the ``\textbf{computational basis}'' of $\mathbb{C}^{|B|}$. The computational basis forms an orthonormal basis of $\mathbb{C}^{|B|}$. 

Given two quantum systems $Q_1$ over $B_1$ and $Q_2$ over $B_2$, we can define a \textbf{product} quantum system $Q_1 \otimes Q_2$ over the set $B_1 \times B_2$. Given $|\phi_1\rangle \in Q_1$ and $|\phi_2\rangle \in Q_2$, we can define the product state $|\phi_1\rangle \otimes |\phi_2\rangle \in Q_1 \otimes Q_2$. 

We say $|\phi\rangle \in Q_1 \otimes Q_2$ is \textbf{entangled} if there does not exist 
$|\phi_1\rangle \in Q_1$ and $|\phi_2\rangle \in Q_2$ such that $|\phi\rangle = |\phi_1\rangle \otimes |\phi_2\rangle$. For example, consider $B_1 = B_2 = \{0,1\}$
and $Q_1 = Q_2 = \mathbb{C}^2$, $|\phi\rangle = \frac{|00\rangle + |11\rangle}{\sqrt{2}}$ is entangled. Otherwise, we say $|\phi\rangle$ is un-entangled. 

A pure state $|\phi\rangle \in Q$ can be manipulated by a unitary transformation $U$. The resulting state $|\phi'\rangle = U |\phi\rangle$. 

We can extract information from a state $|\phi\rangle$ by performing a \textbf{measurement}. A measurement specifies an orthonormal basis, typically the computational basis, and the probability of getting result $x$ is $|\langle x | \phi \rangle|^2$. After the measurement, $|\phi\rangle$ ``collapses'' to the state $|x\rangle$ if the result is $x$. 
                
                For example, given the pure state $|\phi\rangle = \frac{3}{5} |0\rangle + \frac{4}{5} |1\rangle$ measured under $\{|0\rangle ,|1\rangle \}$, with probability $9/25$ the result is $0$ and $|\phi\rangle$ collapses to $|0\rangle$; with probability $16/25$ the result is $1$ and $|\phi\rangle$ collapses to $|1\rangle$.

    We finally assume a quantum computer can  implement any unitary transformation (by using these basic gates, Hadamard, phase, CNOT and $\frac{\pi}{8}$ gates), especially the following two unitary transformations:
        \begin{itemize}
            \item \textbf{Classical Computation:} Given a function $f : X \to Y$, one can implement a unitary $U_f$ over $\mathbb{C}^{|X|\cdot |Y|} \to \mathbb{C}^{|X| \cdot |Y|}$ such that for any $|\phi\rangle = \sum_{x \in X, y \in Y} \alpha_{x, y} |x, y\rangle$, 
            \begin{equation*}
                U_f |\phi\rangle = \sum_{x \in X, y \in Y} \alpha_{x, y} |x, y \oplus f(x)\rangle
            \end{equation*}
            
            Here, $\oplus$ is a commutative group operation defined over $Y$.
            
            \item \textbf{Quantum Fourier Transform:} Let $N = 2^n$. Given a quantum state $|\phi\rangle = \sum_{i=0}^{2^n-1} x_i |i\rangle$, by applying only $O(n^2)$ basic gates, 
	one can compute $|\psi\rangle =  \sum_{i=0}^{2^n-1} y_i |i\rangle$ where the sequence $\{y_i\}_{i=0}^{2^n-1}$ is the sequence achieved by applying the 
	classical Fourier transform ${\sf QFT}_N$ to the sequence $\{x_i\}_{i=0}^{2^n-1}$: 
		\begin{equation*}
			y_k = \frac{1}{\sqrt{N}} \sum_{i=0}^{2^n-1} x_i \omega_n^{i k} 
		\end{equation*} 
	where $\omega_n = e^{2 \pi i / N}$, $i$ is the imaginary unit.

	One interesting property of {\sf QFT} is that by preparing $|0^n\rangle$ and 
	applying ${\sf QFT}_2$ to each qubit, 	$\left({\sf QFT}_2 |0\rangle\right)^{\otimes n} = \frac{1}{\sqrt{2^n}} \sum_{x \in \{0,1\}^n} |x\rangle$ which is a uniform superposition over all possible $x \in \{0,1\}^n$.
        \end{itemize}
    
    For convenience, we sometimes ignore the normalization of a pure state which can be calculated from the context. 

\subsection{Grover's algorithm and BHT algorithm}

\begin{definition}[Database Search Problem]
	Suppose there is a function/database encoded as $F:X \to \{0, 1\}$ and $F^{-1}(1)$  is non-empty. The problem is to find $x^* \in X$ such that $F(x^*) = 1$. 
\end{definition}
We will consider adversaries with quantum access to $F$, meaning they submit queries  as $\sum_{x \in X,y\in\{0,1\}} \alpha_{x,y} |x,y\rangle$ and receive in return $\sum_{x \in X,y\in\{0,1\}} \alpha_{x,y} |x, y\oplus F(x)\rangle$.  Grover's algorithm~\cite{STOC:Grover96} finds a pre-image using an optimal number of queries:
\begin{theorem}[\cite{STOC:Grover96,boyer1998tight}]\label{thm:grover}
	Let $F$ be a function $F: X \to \{0, 1\}$. Let $t = |F^{-1}(1)| > 0$ be the number of pre-images of $1$.  There is a quantum algorithm that finds 
	$x^* \in X$ such that $F(x^*) = 1$ with an expected number of quantum queries to F at most  $O\left(\sqrt{\frac{|X|}{t}}\right)$ even without knowing $t$ in advance. 
\end{theorem}

We will normally think of the number of queries as being fixed, and consider the probability of success given the number of queries.  The algorithm from Theorem~\ref{thm:grover}, when runs for $q$ queries, can be shown to have a success probability $\min(1,O(q^2/(|X|/t)))$.  For the rest of the paper, ``Grover's algorithm'' will refer to this algorithm.

\medskip

Now let us look at another important problem: $2$-collision finding problem on $2$-to-$1$ functions. 
\begin{definition}[Collision Finding on 2-to-1 Functions]
	Assume $|X| = 2 |Y| = 2 N$. Consider a function $F : X \to Y$ such that for every $y \in Y$, $|F^{-1}(y)| = 2$. In other words, every image has exactly two 
pre-images. The problem is to find $x \ne x'$ such that $F(x) = F(x')$. 
\end{definition}

Brassard, H{\o}yer and Tapp proposed a quantum algorithm \cite{BHT98} that solved the problem using only $O(N^{1/3})$ quantum queries. The idea is the following:
\begin{itemize}
\item Prepare a list of  input and output pairs, $L = \{(x_i, y_i=F(x_i)\}_{i=1}^t$ where $x_i$ is drawn uniformly at random and $t = N^{1/3}$; 
\item If there is a $2$-collision in $L$, output that pair. Otherwise,  
\item Run Grover's algorithm on the following function $F'$:  $F'(x) = 1$ if and only if there exists   $i \in \{1, 2, \cdots, t\}$, $F(x) = y_i = F(x_i)$ and $x \ne x_i$. 
	Output the solution $x$, as well as whatever $x_i$ it collides with.
\end{itemize}
This algorithm takes $O(t + \sqrt{N / {t}})$ quantum queries and when $t = \Theta(N^{1/3})$, the algorithm finds a $2$-collision with $O(N^{1/3})$ quantum queries. 

\subsection{Multi-collision Finding and \cite{AC:HosSasXag17}}

Hosoyamada, Sasaki and Xagawa proposed an algorithm for $k$-collision finding on any function $F: X \to Y$ where $|X| \geq k |Y|$ ($k$ is a constant). 
They generalized the idea of \cite{BHT98} and gave the proof for even arbitrary functions. We now briefly talk about their idea. For simplicity in this discussion, we assume $F$ is a $k$-to-$1$ function. 

The algorithm prepares $t$ pairs of $2$-collisions $(x_1, x'_1), \cdots,$ $(x_t, x'_t)$ by running the BHT algorithm $t$ times. If two pairs of $2$-collisions collide, there is at least a $3$-collision (possibly a $4$-collision). 
Otherwise, it uses Grover's algorithm to find a $x'' \ne x_i$, $x'' \ne x'_i$ and $f(x'') = f(x_i) = f(x'_i)$. 
The number of queries is $O(t N^{1/3} + \sqrt{N / t})$. When $t = \Theta(N^{1/9})$, the query complexity is $O(N^{4/9})$. 

By induction,  finding a $(k-1)$-collision requires $O(N^{(3^{k-1}-1)/(2 \cdot 3^{k-1})})$ quantum queries. 
By preparing $t$ $(k-1)$-collisions and applying Grover's algorithm to it, it takes $O(t N^{(3^{k-1}-1)/(2 \cdot 3^{k-1})} + \sqrt{\frac{N}{t}})$ quantum queries to get one 
$k$-collision. It turns out that $t = \Theta(N^{1/3^k})$ and the complexity of finding $k$-collision is $O(N^{(3^{k}-1)/(2 \cdot 3^{k})})$. 

\subsection{Compressed Fourier Oracles and Compressed Phase Oracles}

In~\cite{EPRINT:Zhandry18}, Zhandry showed a new technique for analyzing cryptosystems in the random oracle model.  He also showed that his technique can be used to re-prove several known quantum query lower bounds.  In this work, we will extend his technique in order to prove a new optimal lower bound for multi-collisions.

The basic idea of Zhandry's technique is the following: assume $\mathcal{A}$ is making a query to a random oracle $H$ and the query is $\sum_{x, u, z} a_{x, u, z} |x, u, z\rangle$ where $x$ is the query register, $u$ is the response register and $z$ is its private register.   Instead of only considering the adversary's state $\sum_{x, u, z} a_{x, u, z} |x, u + H(x), z\rangle$ for a random oracle $H$, we can 
actually treat the whole system as 
\begin{equation*}
	\sum_{x, u, z}  \sum_H  a_{x, u, z} |x, u + H(x), z\rangle   \otimes  |H\rangle
\end{equation*}
where $|H\rangle$ is the truth table of $H$. By looking at random oracles that way, Zhandry showed that these five  random oracle models/simulators are 
equivalent: 
\begin{enumerate}
	\item \textbf{Standard Oracles}: 
		\begin{equation*}
			{\sf{StO}} \sum_{x, u, z}  a_{x, u, z} |x, u, z\rangle \otimes \sum_H |H\rangle  \Rightarrow \sum_{x, u, z} \sum_H a_{x, u, z} |x, u + H(x), z\rangle \otimes |H\rangle
		\end{equation*}
	\item \textbf{Phase Oracles}:
		\begin{equation*}
			{\sf{PhO}} \sum_{x, u, z}  a_{x, u, z} |x, u, z\rangle \otimes \sum_H |H\rangle  \Rightarrow \sum_{x, u, z}  a_{x, u, z} |x, u, z\rangle \otimes \sum_H \omega_n^{H(x) \cdot u} |H\rangle
		\end{equation*}
	    where $\omega_n = e^{2 \pi i / N}$ and ${\sf PhO} = (I \otimes {\sf QFT}^\dagger \otimes I) \cdot {\sf StO} \cdot (I \otimes {\sf QFT} \otimes I)$.  In other words, it first applies the {\sf QFT} to the $u$ register, applies the standard query, and then applies ${\sf QFT}^\dagger$ one more time.

	\item \textbf{Fourier Oracles}: We can view $\sum_H |H\rangle$ as ${\sf QFT} |0^N\rangle$. In other words, if we perform Fourier transform 
	on a function that always outputs $0$, we will get a uniform superposition over all the possible functions $\sum_H |H\rangle$. 

	Moreover, $\sum_H \omega^{H(x) \cdot u} |H\rangle$ is equivalent to ${\sf QFT} |0^N \oplus (x, u)\rangle$. Here $\oplus$ means updating (xor) the
		$x$-th entry in the database with $u$. 

	So in this model, we start with $\sum_{x, u, z} a^0_{x, u, z} |x, u, z\rangle \otimes {\sf QFT} |D_0\rangle$ where $D_0$ is an all-zero function. 
	By making the $i$-th query, we have 
		\begin{equation*}
			{\sf PhO} \sum_{x, u, z, D} a^{i-1}_{x, u, z, D} |x, u, z\rangle \otimes {\sf QFT}|D\rangle \Rightarrow \sum_{x,u, z, D} a^{i-1}_{x, u, z, D} |x, u, z\rangle \otimes {\sf QFT}|D \oplus (x, u)\rangle
		\end{equation*}

		The Fourier oracle incorporates ${\sf QFT}$ and operates directly on the $D$ registers:
		\begin{equation*}
			{\sf FourierO} \sum_{x, u, z, D} a^{i-1}_{x, u, z, D} |x, u, z\rangle \otimes |D\rangle \Rightarrow \sum_{x,u,z,D} a^{i-1}_{x, u, z, D} |x, u, z\rangle \otimes |D \oplus (x, u)\rangle
		\end{equation*}

	\item \textbf{Compressed Fourier Oracles}:  The idea is basically the same as Fourier oracles. But  when the algorithm only makes 
		$q$ queries, the database $D$ with non-zero weight contains at most $q$ non-zero entries. 


		So to describe $D$ , we only need at most $q$ different $(x_i, u_i)$ pairs ($u_i \ne 0$) which says the database outputs $u_i$ on $x_i$ and 
			$0$ everywhere else. And $D \oplus (x, u)$ is doing the following: 1) if $x$ is not in the list $D$ and $u \ne 0$, put $(x, u)$ in $D$; 2) 
			if $(x, u')$ is in the list $D$ and $u' \ne u$, update $u'$ to $u' \oplus u$ in $D$; 3) if $(x, u')$ is in the   list and 
			$u' = u$, remove $(x, u')$ from $D$.  

		In the model, we start with $\sum_{x, u, z} a^0_{x, u, z} |x, u, z\rangle \otimes |D_0\rangle$ where $D_0$ is an empty list. After making the 
		$i$-th query, we have 
		\begin{equation*}
			{\sf CFourierO} \sum_{x, u, z, D} a^{i-1}_{x, u, z, D} |x, u, z\rangle \otimes |D\rangle \Rightarrow \sum_{x,u, z, D} a^{i-1}_{x, u, z, D} |x, u, z\rangle \otimes |D \oplus (x, u)\rangle
		\end{equation*}

	\item \textbf{Compressed Standard/Phase Oracles}: 
	These two models are essentially equivalent up to an application of ${\sf QFT}$ applied to the query response register.  From now on we only consider compressed phase oracles. 
	
	By applying {\sf QFT} on the $u$ entries of the database registers of a compressed Fourier oracle, we get a compressed phase oracle. 

		In this model, $D$ contains all the pair $(x_i, u_i)$ which means the oracle outputs $u_i$ on $x_i$ and uniformly at random on other inputs. 
		When making a query on $|x, u, z, D\rangle$,	
		\begin{itemize} 
			\item if $(x, u')$ is in the database $D$ for some $u'$, a phase $\omega_n^{u u'}$ will be added to the state; it corresponds to update $w$ to $w + u$ in the compressed Fourier oracle model where $w = D(x)$ in the compressed Fourier database. 
			
			\item otherwise a superposition is appended to the state $|x\rangle \otimes \sum_{u'} \omega_n^{u u'} |u'\rangle$; it corresponds to put a new pair $(x, u)$ in the list of the compressed Fourier oracle model;
			
			\item also make sure that the list will never have an $(x, 0)$ pair in the compressed Fourier oracle model (in other words, it is $|x\rangle \otimes \sum_{y} |y\rangle$ in the compressed phase oracle model); if there is one, delete that pair;
			
			\item All the `append' and `delete' operations 
			above mean applying {\sf QFT}.  
		\end{itemize}
\end{enumerate}

\section{Algorithm for Multi-collision Finding} \label{algorithm}

In this section, we give an improved algorithm for $k$-collision finding. We use the same idea from~\cite{AC:HosSasXag17} but carefully reorganize the algorithm to reduce the number of queries.

As a warm-up, let us consider the case $k=3$ and the case where $F : X \to Y$ is a $3$-to-$1$ function, $|X| = 3 |Y| = 3 N$.  
They gives an algorithm with $O(N^{4/9})$ quantum queries. 
Here is our algorithm with only $O(N^{3/7})$ quantum queries:
\begin{itemize}
	\item Prepare a list $L = \{(x_i, y_i = F(x_i) )\}_{i=1}^{t_1}$ where $x_i$ are distinct and $t_1 = N^{3/7}$. 
		 This requires $O(N^{3/7})$ classical queries on random points.
	\item Define the following function $F'$ on $X$: 
		\begin{equation*}
				F'(x) = \begin{cases}
						1,  & x \not\in \{x_1, x_2, \cdots, x_{t_1}\} \text{ and } F(x) = y_j \text{ for some } j  \\
						0, &  \text{otherwise}
					\end{cases}
		\end{equation*}
	
		Run Grover's algorithm on function $F'$. Wlog (by reordering $L$), we find $x'_1$ such that $x'_1 \ne x_1$ and 
			$F(x'_1) = F(x_1)$ using $O(\sqrt{N/N^{3/7}}) = O(N^{2/7})$ quantum queries. 
	\item  Repeat the last step $t_2 = N^{1/7}$ times,  we will have 
			$N^{1/7}$ $2$-collisions $L' = \{(x_i, x'_i, y_i)\}_{i=1}^{t_2}$. This takes $O(N^{1/7} \cdot \sqrt{N/N^{3/7}}) = O(N^{3/7})$ quantum queries.
	\item If two elements in $L'$ collide, simply output a $3$-collision. Otherwise, run Grover's on function $G$: 
		\begin{equation*}
				G(x) = \begin{cases}
						1,  & x \not\in \{x_1, x_2, \cdots, x_{t_2}, x'_1, \cdots, x'_{t_2}\} \text{ and } F(x) = y_j \text{ for some } j  \\
						0, &  \text{otherwise}
					\end{cases}
		\end{equation*}
		A $3$-collision will be found when Grover's algorithm finds a pre-image of $1$ on $G$. It takes $O(\sqrt{N/N^{1/7}}) = O(N^{3/7})$ quantum queries.
\end{itemize}
Overall, the algorithm finds a $3$-collision using $O(N^{3/7})$ quantum queries. 

The similar algorithm and analysis works for any constant $k$ and any $k$-to-$1$ function which only requires $O(N^{(2^{k-1}-1)/(2^k-1)})$ quantum queries.
Let $t_1 = N^{(2^{k-1}-1)/(2^k-1)}, t_2 = N^{(2^{k-2}-1)/(2^k-1)}, \cdots$, $ t_i= N^{(2^{k-i}-1)/(2^k-1)}, \cdots, t_{k-1} = N^{1/(2^k-1)}$.
The algorithm works as follows: 
\begin{itemize}
	\item Assume $F : X \to Y$ is a $k$-to-$1$ function and $|X| = k |Y| = k N$. 
	\item Prepare a list $L_1$ of input-output pairs of size $t_1$. With overwhelming probability ($1 - N^{-1/2^k}$), $L_1$ does not contain a collision. 
	By letting $t_0 = N$, this step makes $t_1 \sqrt{N / t_0}$ quantum queries. 
	
	\item Define a function $F_2(x)$ that returns $1$ if the input $x$ is not in $L_1$ but the image $F(x)$ collides with one of the images in $L_1$, otherwise it returns $0$. 
		Run Grover's on $F_2$  $t_2$ times.  Every time Grover's algorithm outputs $x'$, it gives a $2$-collision. 
		With probability $1-O(N^{-1/2^{k}})$ (explained below),  all these $t_2$ collisions do not collide. So we have a list $L_2$ of $t_2$  different $2$-collisions.  
		This step makes $t_2 \sqrt{N / t_1}$ quantum queries. 
	\item For $2 \leq i \leq k-1$, define a function $F_i(x)$ that returns $1$ if the input $x$ is not in $L_{i-1}$ but the image $F(x)$ collides with one of the images of $(i-1)$-collisions in $L_{i-1}$, otherwise it returns $0$. 
		Run Grover's algorithm on $F_i\,$  $t_i$ times.  Every time Grover's algorithm outputs $x'$, it gives an $i$-collision. 
		With probability $1- O(t_i^2/t_{i-1}) = 1 - O(N^{-1/2^k})$,  all these $t_i$ collisions do not collide. So we have a list $L_i$ of $t_i$  different $i$-collisions.  
		This step makes $t_i \sqrt{N / t_{i-1}}$ quantum queries. 
	\item Finally given $t_{k-1}$ $(k-1)$-collisions, using Grover's to find a single $x'$ that makes a $k$-collision with one of the $(k-1)$-collision in  $L_{k-1}$. 
		This step makes $t_k \sqrt{N / t_{k-1}}$ quantum queries by letting $t_k = 1 = N^{(2^{k-k}-1)/(2^k-1)}$. 
\end{itemize}

The number of quantum queries made by the algorithm is simply:
\begin{eqnarray*}
	\sum_{i=0}^{k-1} t_{i+1} \sqrt{N/t_i} &=& \sum_{i=0}^{k-1} \sqrt{N \frac{t^2_{i+1}}{t_i} } \\
					&=&  \sum_{i=0}^{k-1} \sqrt{N \cdot    N^{       \frac{2 \cdot \left(2^{k - (i+1)} - 1 \right) -  (2^{k-i} - 1 )}{2^k - 1}           }      } \\
					&=&  k \cdot N^{(2^{k-1}-1)/(2^k-1)}
\end{eqnarray*}

So we have the following theorem:
\begin{theorem}
	For any constant $k$, any $k$-to-$1$ function $F: X \to Y$ ($|X| = k |Y| = k N$), the algorithm above
	finds a $k$-collision using $O( N^{(2^{k-1}-1)/(2^k-1)})$ quantum queries. 
\end{theorem}

We now show the above conclusion holds for an arbitrary function $F : X \to Y$ as long as $|X| \geq k |Y| = k N$. To prove this, we use the following lemma:

\begin{lemma}
	Let $F : X \to Y$ be a function and  $|X| = k |Y| = k N$. Let $\mu_F = \Pr_x\left[|F^{-1}(F(x))| \geq k \right]$ be the probability that if we choose $x$ uniformly at 
random and $y = F(x)$, the number of pre-images of $y$ is at least $k$. We have $\mu_F \geq \frac{1}{k}$. 
\end{lemma}
\begin{proof}
    We say an input or a collision is good if its image has at least $k$ pre-images. 

	To make the probability as small as possible, we want that if $y$ has less than $k$ pre-images, $y$ should have exactly $k-1$ pre-images. 
So the probability is at least
\begin{equation*}
	\mu_F = \frac{\left|  \{ x \,|\,  x \text{ is good } \}  \right|}{|X|} \geq \frac{k N - (k-1) N}{k N} = \frac{1}{k}
\end{equation*}
\end{proof}

\begin{theorem}
	Let $F : X \to Y$ be a function and  $|X| \geq k |Y| = k N$. The above algorithm finds a $k$-collision using $O( N^{(2^{k-1}-1)/(2^k-1)})$ 
quantum queries with overwhelming probability. 
\end{theorem}

\begin{proof} We prove the case $|X| = k|Y|$.  The case $|X| > k |Y|$ follows readily by choosing an arbitrary subset $X'\subseteq X$ such that $|X'|=k|Y|$ and restrict the algorithm to the domain $X'$.

	As what we did in the previous algorithm, in the list $L_1$, with overwhelming probability, there are $0.999 \mu_F \cdot t_1$ good inputs 
	by Chernoff bound because every input is good with probability $\mu_F$.  Then every $2$-collision in $L_2$ has probability $0.999 \mu_F$ to be good. So by Chernoff bound, 
	$L_2$ contains at least $0.999^2 \mu_F t_2$ good $2$-collisions with overwhelming probability.  
	 By induction, in the final list $L_{k-1}$, with overwhelming probability, there are $0.999^{k-1} \mu_F \cdot t_{k-1}$ 
	good $(k-1)$-collisions. Finally, the algorithm outputs a $k$-collision with probability $1$, by making at most $O(\sqrt{N / (0.99^{k-1} \mu_F t_{k-1})})$ quantum queries.  
	
	As long as $k$ is a constant, the coefficients before $t_i$ are all constants. The number of quantum queries is scaled by a constant and is still $O( N^{(2^{k-1}-1)/(2^k-1)})$ 
	and the algorithm succeeds with overwhelming probability.  
\end{proof}

\section{Lower Bound for Multi-collision Finding} \label{lowerboundsection}

\subsection{Idea in~\cite{EPRINT:Zhandry18}}

We will first show how Zhandry re-proved the lower bound of $2$-collision finding using compressed oracle technique. 
The idea is that when we are working under compressed phase/standard oracle model, a query made by the adversary $(x, u)$ can be 
recorded in the compressed oracle database.

Suppose before making the next quantum query, the current joint state is the following 
\begin{equation*}
	|\phi\rangle = \sum_{x, u, z, D} a_{x, u, z, D} |x, u, z\rangle \otimes |D\rangle
\end{equation*}
where $x$ is the query register, $u$ is the response register, $z$ is the private storage of the adversary and $D$ is the database in the compressed phase oracle model. 
Consider measuring $D$ after running the algorithm.  Because  the algorithm only has information about the points in the database $D$, the only way to have a non-trivial probability of finding a collision is for the $D$ that results from measurement to have a collision. More formally, here is a lemma from \cite{EPRINT:Zhandry18}. 
\begin{lemma}[Lemma 5 from \cite{EPRINT:Zhandry18}]
     Consider a quantum algorithm $A$ making queries to a random oracle $H$ and outputting tuples $(x_1,\cdots,x_k,y_1,\cdots,y_k,z)$. Let $R$ be a collection of such tuples. Suppose with probability $p$, $A$ outputs a tuple such that (1) the tuple is in $R$ and (2) $H(x_i)= y_i$ for all $i$. Now consider running $A$ with compressed standard/phase oracle, and suppose the database $D$ is measured after $A$ produces its output. Let $p'$ be the probability that (1) the tuple is in $R$, and (2) $D(x_i) = y_i$ for all $i$ (and in particular $D(x_i) \ne \bot$). Then
     \begin{equation*}
         \sqrt{p} \leq \sqrt{p'} + \sqrt{k / 2^n}
     \end{equation*}
\end{lemma}
As long as $k$ is small, the difference is negligible. So we can focus on bounding the probability $p'$. 

Let $\tilde{P}_1$ be a projection spanned by all the states with $z, D$ containing at least one collision in the compressed phase oracles. In other words, $z$ contains $x \ne x'$ such that $D(x) \ne \bot$, $D(x') \ne \bot$ and $D(x) = D(x')$. 
\begin{equation*}
	\tilde{P}_1 = \sum_{\substack{x, u, z \\ z, D \text{: $\geq 1$ collision}}} |x, u, z, D\rangle \langle x, u, z, D|  
\end{equation*}
We care about the amplitude (square root of the probability) $\left|\tilde{P}_1  |\phi\rangle \right|$. As in the above lemma, $\left|\tilde{P}_1 |\phi\rangle \right| = \sqrt{p'}$ and  $k = 2$. Moreover, we can bound the amplitude of the following measurement. 
\begin{equation*}
	P_1 = \sum_{\substack{x, u, z \\ D \text{: $\geq 1$ collision}}} |x, u, z, D\rangle \langle x, u, z, D|  
\end{equation*}

Here ``$D:\, \geq 1 \text{ collision}$'' meaning $D$ as a compressed phase oracle, it has a pair of $x \ne x'$ such that $D(x) = D(x')$. 
It is easy to see $|P_1 |\phi\rangle| \geq |\tilde{P}_1 |\phi\rangle|$. So we will focus on bounding $|P_1 |\phi\rangle|$ in the rest of the paper. 

\medskip

For every $|x, u, z, D\rangle$, after making one quantum query, the size of $D$ will increase by at most $1$. 
Let $|\phi_i\rangle$ be the state before making the $(i+1)$-th quantum query and $|\phi'_i\rangle$ be the state after it.  
Let $O$ be the unitary over the joint system corresponding to an oracle query, in other words, $|\phi'_i\rangle = O |\phi_i\rangle$.
By making $q$ queries, the computation looks like the following:
\begin{itemize} 
	\item At the beginning, it has $|\phi_0 \rangle$;  
	\item For $1 \leq i \leq q$, it makes a quantum query; the state $|\phi_{i-1}\rangle$ becomes $|\phi'_{i-1}\rangle$; 
					and it applies a unitary on its registers $U^i \otimes {\sf id}$ to get $|\phi_i\rangle$ where $U^i$ is some unitary defined over the registers $x, u, z$. 
	\item Finally measure it using $P_1$, the probability of finding a collision (in the compressed phase oracle) is at most $|P_1 |\phi_q\rangle |^2$
\end{itemize}

We have the following two lemmas:
\begin{lemma} \label{unitarylemma}
	For any unitary $U^i$, 
		\begin{equation*}
			|P_1 |\phi'_{i-1}\rangle | = |P_1  \cdot  (U^i \otimes {\sf id}) \cdot |\phi'_{i-1} \rangle | = |P_1 |\phi_i \rangle |
		\end{equation*}
\end{lemma}
\begin{proof}
    Intuitively, $P_1$ is a measurement on the oracle's register and $U^i$ is a unitary on the adversary's registers, applying the unitary does not affect the measurement $P_1$. 

    Because $U^i$ is a unitary defined over the registers $x, u, z$ and $P_1$ is a projective measurement defined over the database register $D$, we have 
    \begin{align*}
        \left|P_1 \cdot (U^i \otimes {\sf id }) \cdot |\phi'_{i-1}\rangle \right| &= \left| P_1 \cdot (U^i \otimes {\sf id}) \cdot \sum_{\substack{x, u, z, D}} \alpha_{x, u, z, D} |x, u, z, D\rangle \right| \\
        &= \left| P_1 \cdot (U^i \otimes {\sf id}) \cdot \sum_{\substack{D}} |\psi_D\rangle \otimes |D\rangle \right| \\
        &= \sqrt{\sum_{\text{$D \geq 1 $ collision }} |U^i |\psi_D\rangle|^2}  = \sqrt{\sum_{\text{$D \geq 1 $ collision }} | |\psi_D\rangle|^2}
    \end{align*}
    which is the same as $|P_1 |\phi'_{i-1}\rangle|$.
    
\end{proof}

\begin{lemma} \label{querylemma}
	$\left|P_1  |\phi'_i \rangle \right| \leq \left|  P_1 |\phi_i\rangle \right| + \frac{\sqrt{i}}{\sqrt{N}}$.
\end{lemma}
\begin{proof}

We have 
\begin{eqnarray*}
	|P_1  |\phi'_i\rangle| &=&  |P_1 O |\phi_i\rangle|   \\
			&=&  \left|P_1 O  \left( P_1 |\phi_i\rangle + (I - P_1)  |\phi_i\rangle  \right)  \right| \\
			&\leq&    \left|P_1 O P_1 |\phi_i \rangle \right| + \left| P_1 O (I - P_1)  |\phi_i\rangle    \right|  \\
			&\leq&  \left|P_1 |\phi_i\rangle\right| +  \left| P_1 O (I - P_1)  |\phi_i\rangle    \right|  
\end{eqnarray*}

$|P_1 O P_1 |\phi_i\rangle| \leq |P_1 |\phi_i\rangle| $ is because $P_1 |\phi_i\rangle$ contains only $D$ with collisions. By making one more query, the total magnitude will not increase. 

So we only need to bound the second term $|P_1 O (I - P_1) |\phi_i\rangle|$. 
$(I - P_1) |\phi\rangle$ contains only states $|x, u, z, D\rangle$ that $D$ has no collision. 
If after applying $O$ to a state $|x, u, z, D\rangle$, the size of $D$  does not increase (stays the same or becomes smaller), 
the new database still does not contain any collision. Otherwise, it becomes 
$\sum_{u'} \omega_n^{u u'} |x, u, z, D \oplus (x, u')\rangle$.  And only $|D| \leq i$ out of $N$ possible $D\oplus (x, u')$ contain a collision. 
\begin{eqnarray*}
	\left| P_1 O (I - P_1)  |\phi_i\rangle    \right|    &=&  \left| P_1 O  \sum_{\substack{x, u, z, D \\ D : \text{ no collision}}} a_{x, u, z, D} |x, u, z, D\rangle    \right|  \\
					&=&   \left| P_1 \sum_{\substack{x, u, z, D \\ D : \text{ no collision}}} \frac{1}{\sqrt{N}}  \sum_{u'} \omega_n^{u u' } a_{x, u, z, D} |x, u, z, D \oplus (x, u')\rangle    \right|  \\
					&\leq&    \left( \sum_{\substack{x, u, z, D \\ D : \text{ no collision}}}  \frac{i}{N} \cdot  a^2_{x, u, z, D}    \right)^{1/2} \leq \frac{\sqrt{i}}{\sqrt{N}}  \\
\end{eqnarray*}

\end{proof}

By combining lemma \ref{unitarylemma} and lemma \ref{querylemma}, we have that $|P_1 |\phi_i\rangle| \leq \sum_{j=1}^{i-1} \frac{\sqrt{j}}{\sqrt{N}} = O(i^{3/2} / N^{1/2})$.  So we re-prove the following theorem: 
\begin{theorem}  \label{2collemma}
	For any quantum algorithm, given  a random function $f: X \to Y$ where $|Y| = N$, 
		it needs to make  $\Omega(N^{1/3})$ quantum queries to find a $2$-collision with constant probability.
\end{theorem}

\subsection{Intuition for generalizations}

Here is the intuition for $k=3$: 
as we have seen in the proof for $k=2$, after $T_1 = O(N^{1/3})$
quantum queries, the database has high probability to contain a $2$-collision.  Following the same formula,  after making $T_2$ queries, the amplitude that it contains two $2$-collisions is about  
	\begin{equation*}
		\sum_{T_1 + 1}^{T_2} \frac{\sqrt{i}}{\sqrt{N}} = O\left( \frac{T_2^{3/2} - T_1^{3/2}}{\sqrt{N}}  \right) \,\Rightarrow\,  T_2 = O(2^{2/3} N^{1/3})
	\end{equation*}
And similarly after $T_i = O(i^{2/3} N^{1/3})$, the database will contain $i$ $2$-collisions. Now we just assume between the $(T_{i-1}+1)$-th query and $T_i$-th query, the database contains exactly $(i-1)$ $2$-collisions.

Every time a quantum query is made to a database with $i$  $2$-collisions,  with probability at most $i / N$, the new database will contain a $3$-collision. Similar to the lemma \ref{querylemma}, when we make queries until the database contains  $m$  $2$-collisions, the amplitude that it contains a $3$-collision in the database is at most 
\begin{equation*}
	\sum_{i=1}^m  \frac{\sqrt{i}}{\sqrt{N}} \left(T_i - T_{i-1} \right) \approx  \int_1^m  \frac{x^{1/6}}{N^{1/6}} {\sf d}x \approx x^{7/6} / N^{1/6}
\end{equation*}
which gives us that the number of $2$-collisions is $m = N^{1/7}$. And the total number of quantum queries is $T_m = m^{2/3} \cdot N^{1/3} = N^{3/7}$ which is what we expected.

In the following sections, we will show how to bound the probability/amplitude of finding  a $k=2, 3, 4$-collision and any constant $k$-collision with constant probability. All the proof ideas are explained step by step through the proof for $k=2, 3, 4$. The proof for any constant $k$ is identical to the proof for $k = 4$ but every parameter is replaced with functions of $k$. 

\subsection{Lower Bound for $2$-collisions}

Let $P_{2, j}$ be a projection spanned by all the states with $D$ containing at least $j$ distinct $2$-collisions  in the compressed phase oracle model. 
\begin{eqnarray*}
	P_{2, j}&=& \sum_{\substack{x, u, z \\ D \text{: $\geq j$  $2$-collisions}}} |x, u, z, D\rangle \langle x, u, z, D| 
\end{eqnarray*}

Let the current joint state be $|\phi\rangle$ (after making $i$ quantum queries but before the $(i+1)$-th query), and $|\phi'\rangle$ be the state after making the $(i+1)$-th quantum query. 
\begin{equation*}
	|\phi\rangle = \sum_{x, u, z, D} a_{x, u,  z, D} |x, u\rangle \otimes |z, D\rangle
\end{equation*}
We have the  relation following from lemma \ref{querylemma}: 
\begin{eqnarray*}
		|P_{2, 1} |\phi'\rangle| &\leq&   |P_{2,1} |\phi\rangle| + \frac{\sqrt{i}}{\sqrt{N}} \\
		|P_{2, j} |\phi'\rangle| &\leq&  |P_{2,j} |\phi\rangle| + \frac{\sqrt{i}}{\sqrt{N}} |P_{2, j-1} |\phi'\rangle|    \text{ for all $j > 0$} 
\end{eqnarray*}

Let $|\phi_0\rangle, |\phi_1\rangle, \cdots, |\phi_i\rangle$ be the state after making $0, 1, \cdots, i$ quantum queries respectively. 
Let $f_{i, j} = |P_{2, j} |\phi_i\rangle|$. We rewrite the relations using $f_{i, j}$: 
\begin{eqnarray*}
		f_{i, 1} &\leq&  f_{i-1,1} + \frac{\sqrt{i-1}}{\sqrt{N}} = \sum_{0 \leq l < i} \frac{\sqrt{l}}{\sqrt{N}} < \frac{i^{3/2}}{\sqrt{N}} \\
		f_{i, j} &\leq&  
		f_{i - 1, j} + \frac{\sqrt{i - 1}}{{N}} f_{i - 1, j - 1} \\
		&=& \sum_{0 \leq l_1 < i}\frac{\sqrt{l_1}}{\sqrt{N}} f_{l_1,  j-1}   \\
					&=&  \sum_{0 \leq l_2 < l_1 < i} \frac{\sqrt{l_1}}{\sqrt{N}} \frac{\sqrt{l_2}}{\sqrt{N}} f_{l_2,  j-2}     \\  
					&=& \sum_{0 \leq l_j < l_{j-1} < \cdots < l_2 < l_1 < i} \prod_{k=1}^j \left(  \frac{\sqrt{l_k}}{\sqrt{N}}\right) \\
					&<& \frac{1}{j!} \sum_{0 \leq l_j, l_{j-1}, \cdots, l_2, l_1 < i} \prod_{k=1}^j \left(  \frac{\sqrt{l_k}}{\sqrt{N}}\right) \\
					&=&  \frac{1}{j!} \left( f_{i, 1} \right)^j    \\
					&<&  \left(\frac{e \cdot i^{3/2}}{j \sqrt{N}}\right)^j
\end{eqnarray*}

We observe that when $i = o(j^{2/3} N^{1/3})$, $f_{i, j} = o(1)$. 

\begin{corollary}
For any quantum algorithm, given a random function $f : X \to Y$ where $|Y| = N$, by making $i$ queries, the probability of finding constant $j$ $2$-collisions is at most $O\left(({\frac{i^{3}}{N}})^j\right)$. 
\end{corollary}

\begin{theorem}
	For any quantum algorithm, given a random function $f : X \to Y$ where $|Y| =  N$, it needs to make $\Omega(j^{2/3} N^{1/3})$
quantum queries to find $j$ $2$-collisions with constant probability. 
\end{theorem}

\subsection{Lower Bound for $3$-collisions}
Let $P_{3, k}$ be a projection spanned by all the states with $D$ containing at least $k$ distinct $3$-collisions in the compressed phase model. 
And let $P_{3, j, k}$ be a projection spanned by all the states with $D$ containing \textbf{exactly} $j$ distinct $2$-collisions and at  least $k$ $3$-collisions. 

Let the current joint state be $|\phi\rangle$ (after making $i$ quantum queries but before the $(i+1)$-th query), and $|\phi'\rangle$ be the state after making the $(i+1)$-th quantum query. 
We have the following relation similar to Lemma \ref{querylemma}: 
\begin{align*}
		|P_{3, k} |\phi'\rangle| & \leq |P_{3, k} |\phi\rangle| + \left| P_{3, k} \sum_{l \geq 0} \sum_{\substack{x, u, z \\ D: \text{ exactly $l$ $2$-collisions} \\ \text{exactly $k-1$ $3$-collision} }} \frac{1}{\sqrt{N}} \sum_{u'} \omega_n^{u u'}\cdot \alpha_{x, u, z, D} |x, u, z, D \oplus (x, u')\rangle  \right|
\end{align*}
where the first term means $D$ already contains at least $k$ $3$-collisions before the query; and the second term is the case where a new $3$-collision is added into the database. Similar to Lemma \ref{querylemma}, only $l$ out of $N$ $u'$ will make $D \oplus (x, u')$ contain $k$ $3$-collisions. So we have, 
\begin{align*}
		|P_{3, k} |\phi'\rangle| & \leq |P_{3, k} |\phi\rangle| + \sqrt{ \sum_{l \geq 0} \frac{l}{N} \sum_{\substack{x, u, z \\ D: \text{ exactly $l$ $2$-collisions} \\ \text{exactly $k-1$ $3$-collision} }}  |\alpha|^2_{x, u, z, D}} \\
		& \leq |P_{3, k} |\phi\rangle| + \sqrt{ \sum_{l \geq 0}  \frac{l}{N} |P_{3, l, k - 1} |\phi\rangle|^2 }
\end{align*}

Let $g_{i, k}$ be the amplitude $|P_{3, k} |\phi_i\rangle|$ and 
$g_{i, j, k} = |P_{3, j, k} |\phi_i\rangle|$. 
It is easy to see $g_{i, 0} \leq 1$ for any $i \geq 0$ since it is an amplitude.
We have the following:
\begin{equation*}
	g_{i, k} \leq g_{i-1, k} + \sqrt{\sum_{l \geq 0} \frac{l}{N} \cdot g^2_{i-1, l, k-1}}
\end{equation*}

Let $f_{i, j} = |P_{2, j} |\phi_i\rangle|$. 
Define $h_3(i) = \max\{ 2 e \cdot \frac{i^{3/2}}{\sqrt{N}}, 10 N^{1/8} \}$. We have the following lemma: 
\begin{lemma} \label{recur13}
	\begin{equation*}
		g_{i, k} \leq g_{i - 1, k} + \sqrt{\frac{h_3(i-1)}{{N}}} g_{i - 1, k - 1} + f_{i-1, h_3(i - 1)}
	\end{equation*}
\end{lemma}
\begin{proof}
	\begin{eqnarray*}
			g_{i, k} &\leq& g_{i-1, k} + \sqrt{\sum_{l \geq 0} \frac{l}{N} \cdot g^2_{i-1, l, k-1}}		\\
						&\leq& g_{i -1, k} + \sqrt{\sum_{0 \leq l \leq h_3(i-1)} \frac{l}{{N}} \cdot g^2_{i-1, l, k-1}} + \sqrt{\sum_{l > h_3(i-1)} 1 \cdot g^2_{i-1, l, k-1} }		\\
						&\leq& g_{i -1, k} + \sqrt{\frac{h_3(i-1)}{{N}}} \cdot \sqrt{\sum_{l \geq 0} g^2_{i-1, l, k-1}} + \sqrt{\sum_{l > h_3(i-1)} g^2_{i-1, l, k-1} }		\\	
						&\leq&  g_{i - 1, k} + \sqrt{\frac{h_3(i-1)}{{N}}} \cdot g_{i - 1, k - 1} + f_{i-1, h_3(i-1)}
	\end{eqnarray*}
Here, in the last line, we used the fact that $\sum_{l \geq 0} g^2_{i-1, l, k-1}$ represents the total probability of the database having $k-1$ distinct 3-collisions, and so is equal to $g^2_{i-1,k-1}$.  Similarly, we used that $\sum_{l > h_3(i-1)} g^2_{i-1, l, k-1}$ represents the total probability of having at least $k-1$ distinct 3-collisions \emph{and} at least $h_3(i-1)$ distinct 2-collisions.  This probability is bounded above by the probability of just having at least $h_3(i-1)$ distinct 2-collisions, which is $f^2_{i-1, h_3(i-1)}$.
\end{proof}

\begin{lemma}
	Define $A_i = \sum_{l = 0}^{i-1} \sqrt{\frac{h_3(l)}{N}}$. Then $g_{i, k}$ can be bounded as the following:
	\begin{equation*}
		g_{i, k} \leq \frac{A_i^k}{k!} + 2^{-N^{1/8}} \,\,\,\,\,\,\, \text{ for all } i \leq N^{1/2},  1 \leq k \leq N^{1/8}
	\end{equation*}  
\end{lemma}
\begin{proof}
	If we expand Lemma \ref{recur13}, we have 
	    \begin{align*}
	        g_{i, k} &\leq g_{i - 1, k} + \sqrt{\frac{h_3(i-1)}{{N}}} g_{i - 1, k - 1} + f_{i-1, h_3(i - 1)} \\
	           &\leq g_{i - 2, k} + \sum_{l=i-2}^{i-1}\left( \sqrt{\frac{h_3(l)}{{N}}} g_{l, k - 1} + f_{l, h_3(l)}\right) \\
	           & \vdots \\
	           &\leq g_{0, k} + \sum_{l=0}^{i-1}\left( \sqrt{\frac{h_3(l)}{{N}}} g_{l, k - 1} + f_{l, h_3(l)}\right)
	    \end{align*}
	    where if $k \geq 1$, $g_{0, k} = 0$. Next, 
	
		\begin{eqnarray*}
			g_{i, k} &\leq& \sum_{0 \leq l < i} \sqrt{\frac{h_3(l)}{{N}}} g_{l, k - 1} + \sum_{0 \leq l < i} f_{l, h_3(l)} \\
					&\leq&  \sum_{0 \leq l < i} \sqrt{\frac{h_3(l)}{{N}}} g_{l, k - 1} +  N^{1/2} \left( \frac{1}{2} \right)^{10  N^{1/8} } \\
					&\leq&  \sum_{0 \leq l < i} \sqrt{\frac{h_3(l)}{{N}}} g_{l, k - 1} +   2^{- 9.5 N^{1/8}}
		\end{eqnarray*}

By recursively expanding the inequality, we will get 
        \begin{align*}
            g_{i, k} &\leq \sum_{0 \leq l_1 < i} \sqrt{\frac{h_3(l_1)}{N}} g_{l_1, k - 1} + 2^{-9.5 N^{1/8}} \\
            &\leq \sum_{0 \leq l_1 < i} \sqrt{\frac{h_3(l_1)}{N}} \left( \sum_{0 \leq l_2 < l_1} \sqrt{\frac{h_3(l_2)}{N}} g_{l_2, k - 2} + 2^{-9.5 N^{1/8}} \right) + 2^{-9.5 N^{1/8}} \\
            &\leq \sum_{0 \leq l_1 < i} \sqrt{\frac{h_3(l_1)}{N}} \left( \sum_{0 \leq l_2 < l_1} \sqrt{\frac{h_3(l_2)}{N}} \left( \sum_{0 \leq l_3 < l_2} \sqrt{\frac{h_3(l_3)}{N}}  \cdots \right) + 2^{-9.5 N^{1/8}} \right) + 2^{-9.5 N^{1/8}} \\
            &= \sum_{0 \leq l_k < l_{k-1} < \cdots < l_1 < i} \prod_{j=1}^k \left( \sqrt{\frac{h_3(l_j)}{N}} \right) + 2^{-9.5 N^{1/8}} \sum_{t=0}^{k-1} \sum_{0 \leq l_t < l_{k-1} < \cdots < l_1 < i} \prod_{j=1}^t \left( \sqrt{\frac{h_3(l_j)}{N}} \right) \\
            & < \frac{A_i^k}{k!} + \sum_{t = 0}^{k-1} \frac{A_i^t}{t!} \cdot 2^{-9.5 N^{1/8}} \\
            & < \frac{A_i^k}{k!} + e^{A_i} \cdot 2^{-9.5 N^{1/8}}
        \end{align*}

We then bound $A_i$ for all $i \leq N^{1/2}$ (we can always assume $i = o(N^{1/2})$, because  finding any constant-collision using $O(N^{1/2})$ quantum queries is easy by a quantum computer, just repeatedly applying Grover's algorithm): 
		\begin{eqnarray*}
			A_i &\leq& \sum_{l=1}^i \frac{\sqrt{2 e \cdot l^{3/2}}}{N^{3/4}} + \sum_{l: h_3(l) = 10 N^{1/8}} \frac{\sqrt{10 N^{1/8}}}{N^{1/2}} \\
				&\leq& \sqrt{2 e} \cdot \frac{i^{7/4}}{N^{3/4}} + O\left(N^{-1/48} \right)
		\end{eqnarray*}
Which implies $A_i< 2 e \cdot N^{1/8}$ (by letting $i = \sqrt{N}$). So we complete the proof:
		\begin{eqnarray*}
			g_{i, k}	 &\leq&  \frac{A_i^k}{k!} + e^{A_i } \cdot 2^{-9.5 N^{1/8}}	\\
						 &\leq&   \frac{A_i^k}{k!} + e^{2 e \cdot N^{1/8}} \cdot  2^{- 9.5 N^{1/8}} \\
						 &<&   \frac{A_i^k}{k!} +  2^{- N^{1/8}} 
		\end{eqnarray*}

\end{proof}

\begin{theorem}
	For any quantum algorithm, given a random function $f : X \to Y$ where $|Y| =N$, it needs to make $\Omega(j^{4/7} N^{3/7})$
quantum queries to find $j$ $3$-collisions for any $j \leq N^{1/8}$ with constant probability.
\end{theorem}
\begin{proof}We have two cases:

\begin{itemize}
\item When $j$ is a constant: If $i^* = o(N^{3/7})$, we have 
				$g_{i^*, j} \leq o(1) + O( N^{-1/48} )$. 

\item When $j$ is not a constant:	
For any $j$, let $i^*$ be the largest integer such that $A_{i^*} < \frac{1}{2 e} \cdot j$. In this case, 
			$i^* = O\left(j^{4/7} N^{3/7} \right)$. 
So the probability of having at least $j$ $3$-collisions is bounded by $g^2_{i^*, j}$ where 
$g_{i^*, j} \leq (e A_{i^*} / j)^j + 2^{-N^{1/8}} \leq 2^{- j + 1} + 2^{-N^{1/8}} = o(1)$. 
\end{itemize}		
\end{proof}

\subsection{Lower Bound for $4$-collisions}
Here we show the proof for lower bound of finding $4$-collisions. The proof for arbitrary constant has the same structure but different parameters which is shown in the next section. We prove the case of 4-collisions here to give the idea before generalizing.

Let $f_{i, j}$ be the amplitude of the database containing at least $j$ $3$-collisions after making $i$ quantum queries, 
$g_{i, j, k}$ be the amplitude of the database containing exactly $j$ $3$-collisions and at least $k$ $4$-collisions after $i$ quantum queries, 
$g_{i, k}$ be the amplitude of containing at least $k$ $4$-collisions after $i$ quantum queries. 

As we have seen in the last proof, we have 
\begin{equation*}
	g_{i, k} \leq g_{i-1, k} + \sqrt{\sum_{l \geq 0} \frac{l}{N} g^2_{i - 1, l, k - 1}}
\end{equation*}

Define $h_4(i) = \max\{  (2 e)^{3/2} \cdot \frac{i^{7/4}}{N^{3/4}} , 10 N^{1/16} \}$. Again, we can bound $g_{i, k}$ by dividing the summation into two parts:
\begin{eqnarray*}
	g_{i, k} &\leq& g_{i-1, k} + \sqrt{\sum_{l \leq h_4(i-1)} \frac{l}{N} g^2_{i - 1, l, k - 1}} + \sqrt{\sum_{l > h_4(i-1)} 1 \cdot g^2_{i - 1, l, k - 1}} \\
				&\leq& g_{i-1, k} + \sqrt{\frac{h_4(i-1)}{N}} g_{i-1, k-1} + f_{i - 1, h_4(i-1)} \\
				&\vdots& \\
				&\leq&  \sum_{0 \leq l < i} \sqrt{\frac{h_4(l)}{N}} g_{l, k-1} + \sum_{0 \leq l < i} f_{l, h_4(l)}
\end{eqnarray*}

The second term can be bounded as the following (and we can safely assume $i < N^{1/2}$)
\begin{eqnarray*}
	\sum_{0 \leq l < i} f_{l, h_4(l)}  &\leq&   \sum_{0 \leq l < i} \left(   \frac{A_l^{h_4(l)}}{h_4(l)!} + 2^{-N^{1/8}}  \right) \\
												&\leq& \sum_{0 \leq l < i} \left(  \frac{e A_l}{h_4(l)} \right)^{h_4(l)} + N^{1/2} \cdot 2^{-N^{1/8}} \\
												&\leq&  \sum_{0 \leq l < i} \left( \frac{1}{2} + o(1) \right) ^{10   {N^{1/16}}} + N^{1/2} \cdot 2^{-N^{1/8}} \\
												&\leq& 2^{-9.5 N^{1/16}}
\end{eqnarray*}

Let $B_i = \sum_{0 \leq l < i} \sqrt{\frac{h_4(l)}{N}}$. 
And similarly, for all $i \leq N^{1/2}$,  
\begin{equation*}
	B_i \leq (2 e)^{3/4} \frac{i^{15/8}}{N^{7/8}} + O(N^{- \frac{1}{16} \cdot \frac{1}{14}})
\end{equation*}
The proof follows from the last proof for $k = 3$. A generalized version (for any constant) can be found in the next section. 
And $B_i$ is bounded by $B_{\sqrt{N}}$ which is at most $2 e \cdot N^{1/16}$.

Finally we have the following closed form:
\begin{equation*}
	g_{i, k} \leq \frac{B_i^k}{k!} + \sum_{l=0}^{k-1} \frac{B_i^l}{l!} \cdot 2^{-9.5 N^{1/16}} < \frac{B_i^k}{k!} + e^{B_{\sqrt{N}}} \cdot 2^{-9.5 N^{1/16}} \leq \frac{B_i^k}{k!} + 2^{-N^{1/16}}
\end{equation*}

So we can conclude the following theorem:
\begin{theorem}
	For any quantum algorithm, given a random function $f : X \to Y$ where $N = |Y|$, it needs to make $\Omega(j^{8/15} N^{7/15})$
quantum queries to find $j$ $4$-collisions for any $j \leq N^{1/16}$ with constant probability.
\end{theorem}

\subsection{Lower bound for finding a constant-collision}
In this section, we are going to show that the theorem can be generalized to any constant-collision. 
Let $f_{i, j}$ be the amplitude of the database containing at least $j$ distinct $s$-collisions after $i$ quantum queries, $g_{i, j, k}$ be the amplitude of the database containing exactly $j$ distinct $s$-collisions and at least $k$ distinct $(s+1)$-collisions after $i$ quantum queries. Also let $g_{i, k}$ be the amplitude of the database with at least $k$ distinct $(s+1)$-collisions after $i$ quantum queries. 

We assume $f_{i, j}$ is only defined for $i \leq \sqrt{N}, 1 \leq j \leq N^{1/2^s}$ and $g_{i, k}$ is only defined for $i \leq \sqrt{N}, 1 \leq k \leq N^{1/2^{s+1}}$. It holds for the base cases $s = 4$. 

Define $h_s(i)$ (for any $s \geq 3$) as the following: 
\begin{equation*}
    h_s(i) = \max\left\{  (2e)^{\frac{2^{s-2}-1}{2^{s-3}}} \frac{i^{(2^{s-1}-1)/2^{s-2}}}{N^{(2^{s-2}-1)/2^{s-2}}},  10 \cdot N^{1/2^s}  \right\}
\end{equation*}
It holds for $s = 3, 4$ where $h_3(i) = \max\{(2 e) \cdot i^{3/2} / N^{1/2}, 10 N^{1/8}\}$ and 
$h_4(i) = \max\{ (2 e)^{3/2} \cdot i^{7/4} / N^{3/4}, 10 N^{1/16}\}$. 

Define $A_{i, s} = \sum_{l = 0}^{i-1}  \sqrt{\frac{h_s(l)}{N}}$.
And $A_{i, s} \leq (2e)^{\frac{2^{s-2}-1}{2^{s-2}}} \frac{i^{(2^{s}-1)/2^{s-1}}}{N^{(2^{s-1}-1)/2^{s-1}}} + O(N^{-1/(2^s (2^s -2))})$.
It is easy to see $A_i$ and $B_i$ in the last proof are $A_{i, 3}$ and $A_{i, 4}$. 

\begin{lemma}
     $A_{i, s} \leq (2e)^{\frac{2^{s-2}-1}{2^{s-2}}} \frac{i^{(2^{s}-1)/2^{s-1}}}{N^{(2^{s-1}-1)/2^{s-1}}} + O(N^{-1/(2^s (2^s -2))})$ holds for all constant $s \geq 3$.
\end{lemma}
The lemma is consistent with the cases where $s = 3, 4$. 
\begin{proof}
    \begin{eqnarray*}
        A_{i, s} &=& \sum_{l = 0}^{i-1}  \sqrt{\frac{h_s(l)}{N}} \\
                 &=& \sum_{l : h_s(l) = 10 N^{1/2^s}}  \sqrt{\frac{10 N^{1/2^s}}{N}} +  \sum_{l: h_s(l) > 10 N^{1/2^s}}   \sqrt{\frac{h_s(l)}{N}} \\
                 &=& \sum_{l : h_s(l) = 10 N^{1/2^s}}  \sqrt{\frac{10 N^{1/2^s}}{N}} +  \sum_{l=0}^{i-1}  (2e)^{\frac{2^{s-2}-1}{2^{s-2}}} \frac{l^{(2^{s-1}-1)/2^{s-1}}}{N^{(2^{s-2}-1)/2^{s-1}}} \cdot N^{-1/2}
    \end{eqnarray*}
    where the second summation is at most $(2e)^{\frac{2^{s-2}-1}{2^{s-2}}} \frac{i^{(2^{s}-1)/2^{s-1}}}{N^{(2^{s-1}-1)/2^{s-1}}}$ and the first summation is at most 
    \begin{eqnarray*}
        \sum_{l : h_s(l) = 10 N^{1/2^s}}  \sqrt{\frac{10 N^{1/2^s}}{N}} &=&  \sqrt{\frac{10 N^{1/2^s}}{N}} \cdot O\left( N^{\left( \frac{1}{2^s} + 1 - \frac{1}{2^{s-2}} \right)\cdot \frac{2^{s-2}}{2^{s-1}-1} } \right)   \\
            &\leq& O\left(   N^{-\frac{1}{2} + \frac{1}{2^{s+1}}} \cdot N^{\frac{2^s - 3}{4(2^{s-1}-1)}} \right) \\
            &\leq& O\left(   N^{-\frac{1}{2^s (2^s-2)}} \right) 
    \end{eqnarray*}
    which completes the proof. 
\end{proof}

Finally, we assume $f_{i, j} \leq \frac{A_{i, s}^j}{j!} + O(2^{- N^{1/2^s}})$ which holds for both $s = 3, 4$. We are going to show it holds for $(s+1)$, in other words,     $g_{i, k} \leq\frac{A_{i, s+1}^k}{k!} +  O(2^{- N^{1/2^{s+1}}})$. And by induction, it holds for all constant $s$. 

As we have seen in the last proof, we have the following inequality:
\begin{eqnarray*}
    g_{i, k} &\leq&  g_{i-1, k} + \sqrt{ \sum_{l \geq 0} \frac{l}{N} g^2_{i-1, l, k-1} } \\
            &\leq& g_{i-1, k} + \sqrt{\frac{h_{s+1}(i-1)}{N}} \cdot g_{i-1, k-1} + f_{i-1, h_{s+1}(i-1)}
\end{eqnarray*}
where as $i \leq N^{1/2}$, for sufficient large $N$, the last term $f_{i-1, h_{s+1}(i-1)}$ can be bounded as:
\begin{eqnarray*}
    f_{i-1, h_{s+1}(i-1)} &\leq& \frac{A_{i-1, s}^{h_{s+1}(i-1)}}{h_{s+1}(i-1)!} + O(2^{-N^{1/2^s}}) \\
                &\leq&  \left( e \cdot \frac{ (2e)^{\frac{2^{s-2}-1}{2^{s-2}}} \frac{(i-1)^{(2^{s}-1)/2^{s-1}}}{N^{(2^{s-1}-1)/2^{s-1}}} + O(N^{-1/(2^s (2^s -2))})     }{ \max\left\{  (2e)^{\frac{2^{s-1}-1}{2^{s-2}}} \frac{(i-1)^{(2^{s}-1)/2^{s-1}}}{N^{(2^{s-1}-1)/2^{s-1}}},  10 \cdot N^{1/2^{s+1}}  \right\}}  \right)^{10 N^{1/2^{s+1}}} + O(2^{-N^{1/2^s}}) \\
                &\leq& \left( \frac{1}{2} + o(1) \right)^{10 N^{1/2^{s+1}}} + O(2^{-N^{1/2^s}}) \\
                &<& 2^{- 9.8 N^{1/2^{s+1}}}
\end{eqnarray*}

By expanding the inequality, we get 
\begin{eqnarray*}
    g_{i, k} &\leq&  \sum_{l=0}^{i-1} \sqrt{\frac{h_{s+1}(l)}{N}} g_{l, k-1} + N^{1/2} \cdot 2^{- 9.8 N^{1/2^{s+1}}} \\
    &\leq&  \sum_{l=0}^{i-1} \sqrt{\frac{h_{s+1}(l)}{N}} g_{l, k-1} + 2^{- 9.5 N^{1/2^{s+1}}} \\
    &\leq&  \frac{A_{i, s+1}^k}{k!} + \sum_{l =0}^{k-1} \frac{A_{i, s+1}^l}{l!} \cdot  2^{- 9.5 N^{1/2^{s+1}}} \\
    &\leq&  \frac{A_{i, s+1}^k}{k!} +   e^{A_{i, s+1}} \cdot  2^{- 9.5 N^{1/2^{s+1}}} 
\end{eqnarray*}
Because $i \leq \sqrt{N}$, $A_{i, s+1} < 2 e N^{1/2^{s+1}}$. Finally, we have 
\begin{equation*}
    g_{i, k} \leq\frac{A_{i, s+1}^k}{k!} +  2^{- N^{1/2^{s+1}}} 
\end{equation*}
which completes the induction. So we have the following theorem:
\begin{corollary}
    For any constant $s \geq 2$, let $f_{i,j}$ be the amplitude of the database containing at least $j$ $s$-collisions after $i$ quantum queries. For all $1 \leq j \leq N^{1/2^s}$, we have 
    \begin{equation*}
        f_{i, j} \leq \frac{A_{i, s}^j}{j!} + O(2^{- N^{1/2^s}})
    \end{equation*}
    where 
    \begin{equation*}
        A_{i, s} \leq (2e)^{\frac{2^{s-2}-1}{2^{s-2}}} \frac{i^{(2^{s}-1)/2^{s-1}}}{N^{(2^{s-1}-1)/2^{s-1}}} + O(N^{-1/(2^s (2^s -2))})
    \end{equation*} 
\end{corollary}

\begin{theorem}
	For any quantum algorithm, given a random function $f : X \to Y$ where $N = |Y|$, it needs to make $\Omega(j^{2^{s-1}/(2^s-1)} N^{(2^{s-1}-1)/(2^s-1)})$
quantum queries to find $j$ $s$-collisions for any $j \leq N^{1/2^s}$.

    Moreover, for any quantum algorithm, given a random function $f : X \to Y$ where $N = |Y|$, it needs to make $\Omega(N^{(2^{s-1}-1)/(2^s-1)})$
    quantum queries to find one $s$-collision.
\end{theorem}

\bibliographystyle{alpha}

\bibliography{bib,abbrev3,crypto}

\end{document}